\documentclass[letterpaper,11pt]{article} 
\usepackage{times}  
\usepackage{helvet}  
\usepackage{courier}  
\usepackage[hyphens]{url} 
\usepackage{graphicx} 
\usepackage{natbib}
\usepackage{caption} 
\usepackage{algorithm}
\usepackage{algorithmic}

\usepackage{newfloat}
\usepackage{listings}
\DeclareCaptionStyle{ruled}{labelfont=normalfont,labelsep=colon,strut=off} 
\floatstyle{ruled}
\newfloat{listing}{tb}{lst}{}
\floatname{listing}{Listing}

\usepackage[letterpaper, total={6in, 8in}]{geometry}
\usepackage[colorlinks=true, allcolors=blue!50!black]{hyperref}

\usepackage{amssymb}
\usepackage{xcolor}
\usepackage{amsmath}
\usepackage{amsthm}
\usepackage{xspace}
\usepackage{cleveref}
\usepackage{caption}
\usepackage{subcaption}
\usepackage{tikz}
\usepackage{thm-restate}
\usepackage{natbib}

\definecolor{maincolor}{RGB}{35, 75, 166}
\definecolor{color2}{RGB}{162,173,0}
\definecolor{color3}{RGB}{227,114,34}
\definecolor{color4}{RGB}{202,033,063}

\newcommand{\med}{\ensuremath{\mathrm{med}}}
\newcommand{\mov}{\ensuremath{\mathcal{A}^{\mathcal{F}}}\xspace}
\newcommand{\A}{\ensuremath{\mathcal{A}}\xspace}

\newcommand{\egcitep}[1]{\citep[e.g.,][]{#1}}

\newtheorem{theorem}{Theorem}
\newtheorem{corollary}[theorem]{Corollary}
\newtheorem{lemma}[theorem]{Lemma}

\newtheorem{definition}[theorem]{Definition}
\newtheorem{proposition}[theorem]{Proposition}
\renewcommand{\mu}{\overline{P}}
\allowdisplaybreaks

\date{}

\usepackage{authblk}

\author[1]{\large Rupert Freeman}
\author[2]{\large Ulrike Schmidt-Kraepelin \normalsize}

\affil[1]{\small Darden School of Business, University of Virginia, VA, USA}
\affil[2]{\small TU Eindhoven, The Netherlands}

\title{\fontsize{16}{14}\textbf{Project-Fair and Truthful Mechanisms for Budget Aggregation}}

\usepackage{bibentry}

\begin{document}

\maketitle

\begin{abstract}
We study the budget aggregation problem in which a set of strategic voters must split a finite divisible resource (such as money or time) among a set of competing projects. Our goal is twofold: 
We seek truthful mechanisms that provide fairness guarantees to the projects. 
{For the first objective, we focus on the class of moving phantom mechanisms, which are -- to this day -- essentially the only known truthful mechanisms in this setting.}
For project fairness, we consider the mean division as a fair baseline, and bound the maximum difference between the funding received by any project and this baseline. We propose a novel and simple moving phantom mechanism that provides optimal project fairness guarantees. As a corollary of our results, we show that our new mechanism minimizes the $\ell_1$ distance to the mean for three projects and gives the first non-trivial bounds on this quantity for more than three 
projects.
\end{abstract}

\section{Introduction}
In the budget aggregation problem, a fixed amount of a divisible resource (such as money or time) must be allocated among $m$ competing projects based on the divisions proposed by a set of $n$ voters. For example, in participatory budgeting~\citep{cabannes2004participatory,aziz2021participatory}, citizens vote directly on how a public budget should be divided between a set of public projects. Other examples might include a university department allocating discretionary funding among different initiatives or a group of conference organizers deciding how to divide time among activities such as talks, posters, and social events.

A common and natural solution to this problem is to divide the resource according to the (arithmetic) mean of the votes, guaranteeing that the funding received by each project is proportional to the total support that project receives from the voters. However, using the mean as a budget aggregation rule is not strategyproof.\footnote{We use the terms ``strategyproof'' and ``truthful'' interchangeably.} 
For example, voters can overstate their preference for their favorite projects to bring the funding for that project towards the voter's true preference.

In pursuit of strategyproof mechanisms, \citet{freeman2021truthful} defined the class of \emph{moving phantom mechanisms}, a high-dimensional generalization of the well-known class of generalized median mechanisms for strategyproof aggregation in one dimension~\citep{moulin1980strategy}. Moving phantom mechanisms are strategyproof when voters have disutilities given by the $\ell_1$ distance between their vote and the aggregate division. One particularly natural mechanism, which turns out to be a member of this class, is the one that minimizes the sum of disutilities of the voters~\citep{lindner2008midpoint,goel2019knapsack}. Although this rule can be effective, it can also produce outcomes that differ significantly from the mean (for intuition, consider the median in one dimension). If the mean is considered a desirable outcome, then it would be beneficial to discover strategyproof mechanisms that are more aligned with it.

\citet{freeman2021truthful} introduced the \emph{Independent Markets} mechanism, which is guaranteed to agree with the mean when all voters want to fund only a single project (a mechanism with this property is said to be \emph{proportional}). However, as \citet{caragiannis2022truthful} showed in subsequent work, on other inputs it may produce outcomes that are far from the mean according to the $\ell_1$ distance. They propose a different moving phantom mechanism, the \emph{Piecewise Uniform mechanism}, which never outputs budget divisions that have an $\ell_1$ distance from the mean larger than $\frac{2}{3}$ when there are only three projects.\footnote{The proof relies on solving a nonlinear program, so this bound is subject to a small error.} No positive results are known for higher numbers of projects.

\citet{caragiannis2022truthful} measure the quality of an outcome by its $\ell_1$ distance to the mean. However, this measure does not capture how the deviation from the mean is distributed over the projects. For instance, suppose that the mean division over four projects is given by $(70\%,10\%,10\%,10\%)$ and consider two aggregate divisions: $a = (50\%,30\%,10\%,10\%)$ and $b=(60\%,20\%,0\%,20\%)$. For both aggregates, the $\ell_1$ distance to the mean is $40\%$. However, while in division $a$, the first project is being underfunded relative to the mean by $20\%$ of the budget, in division $b$, no project is over or underfunded by more than $10\%$ of the budget. In this paper, we complement \citeauthor{caragiannis2022truthful}'s approach by studying the $\ell_\infty$ distance to the mean, 
which can be interpreted as a measure of fairness between projects. 
Taking the mean to be a project's ``entitlement,'' by how much does the allocation of any project exceed or fall short of this value?

Given that the projects themselves (or the entities behind them) are typically stakeholders in budget aggregation systems, project-fairness guarantees are important to maintain the confidence of the projects in the system.

\paragraph{Our Contributions.}
We introduce the notion of project fairness for the budget aggregation problem. While our definition is technically similar to the proportionality measure of \citet{caragiannis2022truthful} in that we are interested in the worst case distance (according to some metric) from the mean, the two metrics can differ substantially in which outcomes they prefer. That said, they are related in that an upper bound on the $\ell_\infty$ distance implies an upper bound on the $\ell_1$ distance and vice versa; see \Cref{sec:l1}.

We focus on project fairness for the class of moving phantom mechanisms. Whether there exist (anonymous, neutral, and continuous) strategyproof mechanisms outside of this class remains an intriguing open question. As our main result, we define the \emph{Ladder mechanism}, a new moving phantom mechanism that is guaranteed to output a budget division with $\ell_\infty$ distance from the mean equal to at most $\frac{1}{2}-\frac{1}{2m}$. This bound is tight for moving phantom mechanisms. We additionally show that, while our mechanism may underfund a project by this amount relative to the mean, it never overfunds a project by more than $1/4$, a property that we show to be common to all proportional mechanisms.

As a corollary of our result, we show that our new mechanism guarantees an $\ell_1$ distance from the mean of no more than $\frac{2}{3}$ for instances with three projects, which matches the known lower bound. This closes a (very small) gap that was left open by \citet{caragiannis2022truthful}, who obtained an upper bound of $\frac{2}{3} + 10^{-5}$ by a complex proof that involved characterizing worst case instances and then solving a non-linear program. In contrast, our proof is combinatorial and relatively simple in comparison. We additionally obtain non-trivial bounds on the $\ell_1$ distance from the mean for 4, 5, and 6 projects. Prior to our work, no mechanisms were known to guarantee an $\ell_1$ distance less than a trivial upper bound for more than 3 projects.

\paragraph{Related Work.}
Portioning, also known as (unbounded) divisible participatory budgeting, is an umbrella term for problems in which a continuous divisible resource must be divided among alternatives. The budget aggregation problem is an example of portioning where voters submit complete budget allocation proposals; in addition to the papers discussed above, \citet{elkind2023settling} perform an axiomatic analysis of various rules in this setting, and find that the mean performs well relative to the other rules they consider. In particular, it is the only one of the considered rules to satisfy the score representation axiom, a natural proportionality property. \citet{goyal2023low} study mechanisms with low sample complexity in terms of their social welfare approximation guarantees.
Other variants of portioning include voters submitting ordinal preferences~\citep{airiau2023portioning}, dichotomous preferences~\egcitep{bogomolnaia2005collective,brandl2021distribution,michorzewski2020price}, or more general cardinal utility functions over alternatives~\citep{fain2016core,wagner2023strategy}. For an overview of other models and additional related work in participatory budgeting, we refer to the survey of~\citet{aziz2021participatory}.

For the special case of two projects,\footnote{Since we have a normalization constraint, the two-project case has only one degree of freedom.} moving phantom mechanisms reduce to the generalized median mechanisms of~\citet{moulin1980strategy}, which take the median of $n+1$ fixed ``phantom'' votes and the $n$ submitted votes. These mechanisms have been extensively studied, most notably in the context of strategyproof facility location~\citep{procaccia2013approximate,aziz2021strategyproof}. Connections between generalized median mechanisms and mean approximation in one dimension have also been made previously in various contexts~\citep{renault2005protecting,renault2011assessing,caragiannis2016truthful,jennings2023new,caragiannis2022truthful}. All of these papers identify the uniform phantom mechanism, which places phantom votes at uniform intervals of $1/n$, as the most desirable generalized median 
from this perspective. As with other proportional moving phantom mechanisms in the literature, the Ladder mechanism draws heavy inspiration thereof. 

Alternative multidimensional aggregation settings that do not require votes and outcomes to sum to one exist in the literature~\egcitep{barbera1993generalized,barbera1997voting,border1983straightforward,peters1992pareto}. Typically, strategyproof mechanisms in these models can be decomposed into one-dimensional mechanisms taking a generalized median in every coordinate, which would violate our normalization requirement. Accordingly, these problems are very different to ours from a technical perspective.

\section{Preliminaries}
\label{sec:prelim}

{For any $k \in \mathbb{N}$, let $[k] = \{1,\dots,k\}$ and $[k]_0 = \{0,1,\dots,k\}$. We denote by $N = [n]$ the set of voters and by $M=[m]$ the set of projects. 
For any $m \in \mathbb{N}$, we define $\Delta^{(m)} = \{q \in [0,1]^{m} \mid \sum_{j \in [m]} q_j=1\}$ to be the standard simplex. For a set of projects $M$, each voter indicates their \emph{ideal} budget distribution over the projects, i.e., an element of $\Delta^{(m)}$. Formally, these preferences are summarized in a preference profile $P$, which is a matrix $P \in [0,1]^{n \times m}$ with $(P_{ij})_{j \in [m]} \in \Delta^{(m)}$ for every $i \in N$. A budget aggregation mechanism $\mathcal{A}$ takes as input a preference profile $P$ 
and outputs an element from $\Delta^{(m)}$. For a given preference profile $P$, and for any $j \in [m]$, let $\overline{P}_j = \frac{1}{n}\sum_{i \in [n]}P_{ij}$ be the average support of this project.

\smallskip
\textbf{Moving Phantom Mechanisms} 
For $n \in \mathbb{N}$, a \emph{phantom system} $\mathcal{F}_n=\{f_k : k \in [n]_0\}$ is a family of functions, where $f_k : [0,1] \rightarrow [0,1]$ is a continuous, non-decreasing function with $f_{k}(0) = 0$ and 
$f_k(1)= 1$
for each $k$, and $f_0(t) \geq f_1(t) \geq \dots \geq f_n(t)$ for all $t \in [0,1]$. Then, for any preference profile $P$, let $t^* \in [0,1]$ be chosen such that 
$$
\sum_{j \in [m]}\med(f_0(t^*),\dots,f_n(t^*),P_{1j},\dots, P_{nj})=1,
$$ 
where ``med'' is the median. Then, we define $$\mathcal{A}^{\mathcal{F}_n}(P)_j = \med(f_0(t^*),\dots,f_n(t^*),P_{1j},\dots, P_{nj})$$ 
and say that $\mathcal{A}^{\mathcal{F}_n}$ \emph{reaches normalization at} $t^*$. While $t^*$ is not always unique, the resulting budget allocation is unique. 

Since phantom systems are defined for fixed $n \in \mathbb{N}$, we are interested in \emph{families} of phantom systems, $\mathcal{F} = \{\mathcal{F}_n \mid n \in \mathbb{N}\}$, and define the \emph{moving phantom mechanism} \mov by applying mechanism $\mathcal{A}^{\mathcal{F}_n}$ to any profile with $n$ voters.

Following~\citet{freeman2021truthful}, we will pictorially represent (snapshots of) moving phantom mechanisms in the following way (see, for example, Figure~\ref{fig:first-example}). Projects are represented by vertical bars, with voter reports indicated by black horizontal line segments. The vertical position of the segment indicates the report $P_{ij}$. Phantom positions are indicated by solid blue lines. On every project, the median of the voter and phantom positions is indicated by a rectangle.

 \smallskip

\textbf{Strategyproofness} For any $q \in \Delta^{(m)}$, the disutility of a voter $i$ is assumed to be the $\ell_1$ distance from $q$ to its ideal point, i.e., $\sum_{j \in [m]}|P_{ij}-q_j|$. 
It is known that all moving phantom mechanisms are strategyproof in the sense that no voter can decrease the $\ell_1$ distance from the aggregate to their ideal distribution by reporting a distribution that is not their ideal one~\citep{freeman2021truthful}. Note that strategyproofness of moving phantom mechanisms rests crucially on the assumption of $\ell_1$ (dis)utilities. We refer the reader to \citet{nehring2019resource} and \citet{goel2019knapsack} for natural interpretations of this utility model in the budgeting setting (and to \citet{varloot2022level} for a setting where a utility model other than $\ell_1$ is more appropriate).

\smallskip
\textbf{Proportionality} We say that a voter $i \in N$ is \emph{single minded} if $P_{ij} \in \{0,1\}$ for all $j \in M$. \citet{freeman2021truthful} define a budget aggregation mechanism to be \emph{proportional} if, for any profile consisting of single-minded voters only, it holds that $\A(P)_j = \overline{P}_j$ for all $j \in M$.
\smallskip

We now introduce the main novel concept of the paper, i.e., a guarantee for the maximum deviation of the funding received by any project from the funding given to this project by the mean aggregation function (which is not strategyproof). Since these bounds can be made more precise by parameterizing them by $n$ and $m$, we express the resulting bounds in terms of a function $\alpha(n,m)$. 

\smallskip

\textbf{Project Fairness}
For a function $\alpha: \mathbb{N} \times \mathbb{N}\rightarrow \mathbb{R}$, we say that a budget aggregation mechanism is $\alpha$-project fair if, for any preference profile $P$ on $n$ voters and $m$ projects, and any $j \in [m]$, it holds that $|\A(P)_j - \overline{P}_j| \leq \alpha(n,m)$. In addition, we say that a mechanism overfunds by at most $\alpha$ if $\A(P)_j - \overline{P}_j \leq \alpha(n,m)$ for all $j \in [m]$, and it underfunds by at most $\alpha$ if $\overline{P}_j - \A(P)_j \leq \alpha(n,m)$ for all $j \in [m]$. Clearly, a mechanism is $\alpha$-project fair if and only if it overfunds by at most $\alpha$ and underfunds by at most $\alpha$. For simplicity, any function $\alpha$ in this paper maps from $\mathbb{N} \times \mathbb{N}$ to $\mathbb{R}$. 
}

\section{Lower Bounds}

In this section, we provide several lower bounds on the $\alpha$-project fairness for (subclasses of) moving phantom mechanisms. This paves the way to the introduction of a new mechanism guaranteeing optimal project fairness. 
We say that a budget aggregation mechanism is \textit{zero unanimous} if it never funds a project that every voter agrees should receive zero funding, i.e., $P_{ij}=0$ for all $i \in [n]$ implies that $\A(P)_j = 0$.
Zero unanimity is a restriction of the score unanimity condition of \citet{elkind2023settling}, which says that whenever all
agents unanimously agree on the funding for a particular project, then
this project should receive exactly that level of funding.
In \Cref{lem:overfund-upperbound}, we start by providing a lower bound on the overfunding guarantee of any zero-unanimous moving phantom mechanism.\footnote{
Dropping zero unanimity allows for moving phantom mechanisms with better overfunding guarantees. 
For example, the mechanism that always outputs $\frac{1}{m}$ for every project never overfunds by more than $\frac{1}{m}$. That said, this mechanism clearly suffers from high underfunding and completely ignores the voters' preferences.}

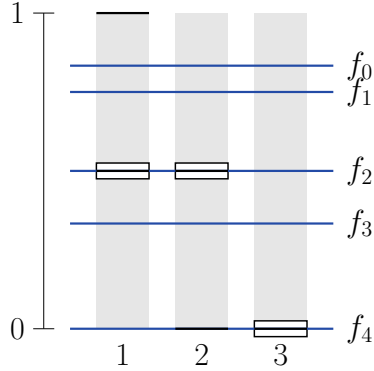
\begin{figure}
    \centering
    \scalebox{0.7}{
    \begin{tikzpicture}
        \draw (0,0) -- (0,6);  
        \draw (-.2,0) -- (0.2,0); 
        \draw (-.2,6) -- (0.2,6); 
        \node at (-0.5,0) {\LARGE$0$}; 
        \node at (-0.5,6) {\LARGE$1$}; 
        \filldraw[fill=black!10,draw=none] (1,0) rectangle (2,6); 
        \filldraw[fill=black!10,draw=none] (2.5,0) rectangle (3.5,6); 
        \filldraw[fill=black!10,draw=none] (4,0) rectangle (5,6); 
        \node at (1.5,-0.5) {\LARGE$1$};
        \node at (3,-0.5) {\LARGE$2$};
        \node at (4.5,-0.5) {\LARGE$3$};

        \draw[fill=white,thick] (1,2.85) rectangle (2,3.15); 
        \draw[fill=white,thick] (2.5,2.85) rectangle (3.5,3.15); 
        \draw[fill=white,thick] (4,-.15) rectangle (5,.15); 

        \draw[very thick,maincolor] (0.5,5) -- (5.5,5); 
        \node at (6,5) {\LARGE$f_0$}; 
        \draw[very thick,maincolor] (0.5,4.5) -- (5.5,4.5); 
        \node at (6,4.5) {\LARGE$f_1$}; 
        \draw[very thick,maincolor] (0.5,2) -- (5.5,2); 
        \node at (6,2) {\LARGE$f_3$}; 
        \draw[very thick,maincolor] (0.5,0) -- (5.5,0); 
        \node at (6,0) {\LARGE$f_4$}; 
    
        \draw[very thick,maincolor] (0.5,3) -- (5.5,3); 
        \node at (6,3) {\LARGE$f_2$}; 

        \draw[very thick] (1,6) -- (2,6); 
        \draw[very thick] (1,3) -- (2,3); 
        \draw[very thick] (2.5,0) -- (3.5,0); 
        \draw[very thick] (2.5,3) -- (3.5,3); 
        \draw[very thick] (4,0) -- (5,0); 
        
    \end{tikzpicture}}
    \caption{Example from \Cref{lem:overfund-upperbound} for $n=4$ and $m=3$. 
    See~\Cref{sec:prelim}
    for an explanation of how to read our figures. Note that the black line segments each represent two voters who both make the same report.}
    \label{fig:first-example}
\end{figure}

\begin{proposition} \label{lem:overfund-upperbound}
    Let \mov be a zero-unanimous moving phantom mechanism. Then, there exists no $\alpha$ satisfying 
    \begin{itemize}
        \item $\alpha(n,m)<\frac{1}{4}$ for any $n,m \in \mathbb{N}$, where $n$ is even, or 
        \item $\alpha(n,m)<\frac{1}{4}\Big(1-\frac{1}{n}\Big)$ for any $n,m \in \mathbb{N}$, where $n$ is odd, 
    \end{itemize}
    such that \mov overfunds by at most $\alpha$.
\end{proposition}

\begin{proof}
    For any $n,m \in \mathbb{N}$, consider the instance in which the voters in $N_1 = [\lfloor n/2 \rfloor]$ report $P_{i1}=1$ and $P_{ij}=0$ for all $j\in M \setminus \{1\}$ and the voters in $N \setminus N_1$ report $P_{i1}=P_{i2}=1/2$ and $P_{ij}=0$ $j\in M \setminus \{1,2\}$.
    By zero unanimity, $\mov(P)_j = 0$ for all $j \in M\setminus\{1,2\}$, and therefore $\mov$ reaches normalization when $f_{\lfloor n/2 \rfloor}(t)=1/2$, returning the budget distribution $\mov(P)_1=\mov(P)_2=1/2$.
    We refer to \Cref{fig:first-example} for an illustration of the case $n=4$ and $m=3$. 
    Now, when $n$ is even, it holds that $\overline{P}_2 = \frac{1}{4}$ and therefore $\mov(P)_2 - \overline{P}_2 = 1/4$. When $n$ is odd, we get that $\overline{P}_2 = \frac{1}{4}\big(1 + \frac{1}{n}\big)$ and therefore $\mov(P)_2 - \overline{P}_2 = \frac{1}{4}\big(1 - \frac{1}{n}\big)$. 
\end{proof}

We continue by providing a lower bound on the underfunding guarantee that any moving phantom mechanism can provide. To this end, we use an example provided by \citet{caragiannis2022truthful}.

\begin{proposition} \label{lem:underfund-lowerbound}
        Let \mov be a moving phantom mechanism. Then, there exists no $\alpha$ such that 
        \begin{itemize}
            \item $\alpha(n,m)<\frac{1}{2}\Big( 1 - \frac{1}{m}\Big) \text{ for any } n,m \in \mathbb{N}$, $n$ even, or 
            \item 
            $\alpha(n,m)<\frac{1}{2}\Big(1-\frac{1}{m}\Big)\Big(1-\frac{1}{n}\Big) \text{ for } n,m \in \mathbb{N}$, $n$ odd, 
        \end{itemize}
        and \mov underfunds by at most $\alpha$. 
\end{proposition}

\begin{proof}
For any $n,m \in \mathbb{N}$, let $N_1 = [\lfloor n/2 \rfloor]$ and $N_2 = N \setminus N_1$. Then, define the profile $$P_{ij}= \begin{cases} 1 & \text{ if } i \in N_1, j = 1 \\ 0 & \text{ if } i \in N_1, j\neq 1 \\  1/m & \text{ if } i \in N_2.\end{cases}$$ \citet[Theorem 7]{caragiannis2022truthful} prove 
that for this profile under the restriction that $n$ is even, any moving phantom mechanism returns $\mov(P)_j = 1/m$ for all $j \in [m]$. It is easy to verify by going through their arguments that the same holds when $n$ is odd. Hence, we receive the following lower bounds for the underfunding guarantees: For $n$ even, it holds that $\overline{P}_1 = \frac{1}{2}+\frac{1}{2m}$, which implies $\overline{P}_1 - \mov(P)_1 = \frac{1}{2} - \frac{1}{2m}.$ For odd $n$, we get that $\overline{P}_1 = \frac{1}{2}\big(1+\frac{1}{m}\big(1 + \frac{1}{n}\big)-\frac{1}{n}\big)$ and therefore $\overline{P}_1 - \mov(P)_1 = \frac{1}{2}\big(1-\frac{1}{m}\big(1-\frac{1}{n}\big) - \frac{1}{n}\big)$.
\end{proof}

In \Cref{lem:upwards-deviation-egal} of the next section, we show that any proportional mechanism overfunds by at most $\frac{1}{4}$. Hence, we can focus on finding a proportional mechanism with optimal underfunding guarantee. While doing so, we first seek to understand the space of mechanisms that are optimal for large $m$, i.e., mechanisms with an underfunding guarantee $\alpha$ satisfying $\lim_{m \rightarrow \infty} \alpha(n,m) = \frac{1}{2}$.
In \Cref{lem:lowerbound-independent-type}, we exhibit a class of moving phantom mechanisms that do not provide an optimal asymptotic underfunding guarantee. As we show in \Cref{cor:lowerbound-independent}, this class includes the \emph{Independent Markets mechanism}, which has been previously studied by \citet{freeman2021truthful} and \citet{caragiannis2022truthful}. Intuitively, \Cref{lem:lowerbound-independent-type} implies that any mechanism that moves a phantom with high index while the symmetric phantom of low index is still low has to have a higher asymptotic underfunding guarantee than $\frac{1}{2}$.

\begin{proposition} \label{lem:lowerbound-independent-type}
   Let $\mov$ be a moving phantom mechanism and $k \in [\lfloor n/2\rfloor]_0$. Then, for any $t \in [0,1]$ such that $f_{n-k}(t) >0$, there exists no $\alpha$ such that
    $$\lim_{m \rightarrow \infty} \alpha(n,m) \leq \frac{n-k}{n} - f_k(t)$$
    and  \mov 
    underfunds by at most $\alpha$.
\end{proposition}

\begin{proof}
    Let \mov, $k$, and $t$ be as in the proposition assumptions. 
    Now, for any $m \in \mathbb{N}$ such that $f_{k}(t) + (m-1) \cdot f_{n-k}(t) > 1$ we can construct a simple instance in which $n-k$ voters cast the vote $(1,0,\dots,0)$ and the remaining $k$ voters cast the vote $(0,\frac{1}{m-1},\dots,\frac{1}{m-1})$. By construction, the mechanism \mov is normalized for some $t' < t$. As a result, we get that $$\overline{P}_1 - \mov(P)_1 \geq \frac{n-k}{n} - f_k(t),$$ 
    and the proposition statement follows.
\end{proof}

Below, we show the implication for the \emph{Independent Markets} mechanism, which is defined by the phantom system\footnote{This does not exactly fit the definition of a moving phantom mechanism since $f_k(1)<1$ for $k \in [n]$. However, it is known that $f_k(1) \ge 1-\frac{k}{n}$ is sufficient to always achieve normalization~\citep{freeman2021truthful}, thus moving all phantoms to 1 is redundant.} $$\mathcal{F}_n = \Big\{f_k(t) = t\cdot \frac{n-k}{n} \text{ for all } k \in [n]_{0}, t \in [0,1]\Big\},$$ for all $n \in \mathbb{N}$.

\begin{corollary} \label{cor:lowerbound-independent}
    For the Independent Markets mechanism and any $\epsilon >0$, there exists no function $\alpha$ satisfying $$\lim_{m \rightarrow \infty}\alpha(n,m) \leq (1-\epsilon)\frac{n-1}{n}$$ 
    such that the Independent Markets mechanism underfunds by at most $\alpha$. 
\end{corollary}

\begin{proof}
    For any $\epsilon>0$ it holds that $f_1(\epsilon) = \epsilon \frac{n-1}{n}$ and $f_{n-1}(\epsilon) = \epsilon \frac{1}{n}>0$. Therefore,  \Cref{lem:lowerbound-independent-type} implies that there exists no $\alpha$ with $\lim_{m \rightarrow \infty} \alpha(n,m) \leq \frac{n-1}{n} - \epsilon \frac{n-1}{n}$, such that the Independent Markets mechanism underfunds by at most $\alpha$. 
\end{proof}

\section{The Ladder Mechanism}

\Cref{lem:lowerbound-independent-type} narrows down the space of moving phantom mechanisms that can achieve a project fairness guarantee of $\frac{1}{2}$ in the limit as $m$ grows: At any moment in time $t \in [0,1]$ when $f_{n-k}(t)>0$ for any $k \in [\lfloor n/2\rfloor]_0$, it needs to hold that $f_{k}(t)\geq \frac{1}{2}-\frac{k}{n}$. Thus, we aim to construct mechanisms that first move the upper phantoms while keeping the lower phantoms at zero. However, while doing so we have to be careful. For example, it might be tempting to consider the moving phantom mechanism which starts by increasing $f_0$ from $0$ to $1$ (while keeping all other phantoms at $0$), then moves $f_1$ from $0$ to $\frac{n-1}{n}$, and so on. However, the large gap between the  middle phantoms leads to problems itself: For any odd $n$, there exists a profile\footnote{Let $n\in \mathbb{N}$ be odd, $m=3$. Then, $\lfloor n/2 \rfloor$ of the voters report $(1,0,0)$ while $\lceil n/2 \rceil$ of the voters report $(0,\frac{1}{2},\frac{1}{2})$. The mechanism described above would output $(0,\frac{1}{2},\frac{1}{2})$ and therefore underfund project $1$ by $\frac{1}{2}-\frac{1}{2n}$.
The same profile provides a counter example for the mechanism maximizing utilitarian social welfare \cite{freeman2021truthful}.} with $m=3$ in which this algorithm underfunds a project by $\frac{1}{2}-\frac{1}{2n}$, which 
is larger than the lower bound from \Cref{lem:underfund-lowerbound}.

There exists one moving phantom mechanism in the literature that avoids both of the described issues: \citet{caragiannis2022truthful} introduced the \emph{Piecewise Uniform mechanism}, which 
in a first phase spreads the upper $\lceil \frac{n+1}{2} \rceil$ phantoms uniformly within the interval $[0,1]$, and, in a second phase, spreads the lower $\lfloor \frac{n+1}{2} \rfloor$ phantoms uniformly within the interval $[0,\frac{1}{2}]$ while pushing the first half of the phantoms into the interval $[\frac{1}{2},1]$. The Piecewise Uniform mechanism avoids the issue captured by \Cref{lem:lowerbound-independent-type} and at 
no time creates too large a gap between consecutive phantoms.
Hence, the mechanism is in fact a promising candidate for optimal project fairness. That said, the precise definition of the mechanism is intricate, making it difficult to analyze. Instead, we propose a novel and arguably simpler mechanism which also avoids both issues and additionally allows for an elegant proof of optimal project fairness. 

We refer to our mechanism as the \emph{Ladder mechanism}, and there are two ways to gain intuition for it: The first view, giving the mechanism its name, thinks of a rope ladder where the ladder rungs correspond to the phantoms. The ladder is then pulled up by its top rung. The second view, being closer to its formal definition, imagines the phantoms being uniformly spread within the interval $[-1,0]$, and then, as $t$ increases, being pushed upwards (with equal speed) until they are uniformly spread in $[0,1]$. However, since phantoms need to be non-negative, they only become ``active'' once they cross $0$, which is ensured by the max function in the following definition.

\begin{definition}[Ladder Mechanism]
    The Ladder mechanism is the moving phantom mechanism defined by the following phantom system for any $n \in \mathbb{N}$: 
    $$f_k(t) = \max\Big(t - \frac{k}{n},0\Big) \text{ for all } k \in [n]_0, t \in [0,1].$$
\end{definition}

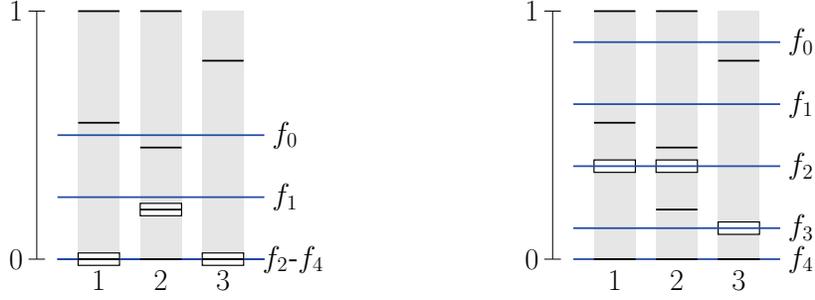
\begin{figure}
    \centering \hspace*{2cm}
    \begin{subfigure}{0.45\textwidth}
    \scalebox{0.55}{\centering
    \begin{tikzpicture}
        \draw (0,0) -- (0,6);  
        \draw (-.2,0) -- (0.2,0); 
        \draw (-.2,6) -- (0.2,6); 
        \node at (-0.5,0) {\huge$0$}; 
        \node at (-0.5,6) {\huge$1$}; 
        \filldraw[fill=black!10,draw=none] (1,0) rectangle (2,6); 
        \filldraw[fill=black!10,draw=none] (2.5,0) rectangle (3.5,6); 
        \filldraw[fill=black!10,draw=none] (4,0) rectangle (5,6); 
        \node at (1.5,-0.5) {\huge $1$};
        \node at (3,-0.5) {\huge$2$};
        \node at (4.5,-0.5) {\huge$3$};

        \foreach \d/\e/\f in {0/0.2/0}{
        \draw[fill=white,thick] (1,\d*6-0.15) rectangle (2,\d*6+0.15); 
        \draw[fill=white,thick] (2.5,\e*6-0.15) rectangle (3.5,\e*6+0.15); 
        \draw[fill=white,thick] (4,\f*6-0.15) rectangle (5,\f*6+0.15); 
        }

        \draw[very thick,maincolor] (0.5,3) -- (5.5,3); 
        \node at (6,3) {\huge$f_0$}; 
        \draw[very thick,maincolor] (0.5,1.5) -- (5.5,1.5); 
        \node at (6,1.5) {\huge$f_1$}; 
        \draw[very thick,maincolor] (0.5,0) -- (5.5,0);         
        \draw[very thick,maincolor] (0.5,0) -- (5.5,0); 
        \draw[very thick,maincolor] (0.5,0) -- (5.5,0); 
        \node at (6.2,0) {\huge$f_2$-$f_4$}; 

        \foreach \a/\b/\c in {0/0.2/0.8, 1/0/0,0/1/0,0.55/0.45/0}{
        \draw[very thick] (1,\a*6) -- (2,\a*6); 
        \draw[very thick] (2.5,\b*6) -- (3.5,\b*6); 
        \draw[very thick] (4,\c*6) -- (5,\c*6); 
        }
        
    \end{tikzpicture}}
    \end{subfigure}\begin{subfigure}{0.45\textwidth}
        \scalebox{0.55}{\centering
            \begin{tikzpicture}
        \draw (0,0) -- (0,6);  
        \draw (-.2,0) -- (0.2,0); 
        \draw (-.2,6) -- (0.2,6); 
        \node at (-0.5,0) {\huge$0$}; 
        \node at (-0.5,6) {\huge$1$}; 
        \filldraw[fill=black!10,draw=none] (1,0) rectangle (2,6); 
        \filldraw[fill=black!10,draw=none] (2.5,0) rectangle (3.5,6); 
        \filldraw[fill=black!10,draw=none] (4,0) rectangle (5,6); 
        \node at (1.5,-0.5) {\huge$1$};
        \node at (3,-0.5) {\huge$2$};
        \node at (4.5,-0.5) {\huge$3$};

        \foreach \d/\e/\f in {0.375/0.375/0.125}{
        \draw[fill=white,thick] (1,\d*6-0.15) rectangle (2,\d*6+0.15); 
        \draw[fill=white,thick] (2.5,\e*6-0.15) rectangle (3.5,\e*6+0.15); 
        \draw[fill=white,thick] (4,\f*6-0.15) rectangle (5,\f*6+0.15); 
        }
        \draw[very thick,maincolor] (0.5,5.25) -- (5.5,5.25); 
        \node at (6,5.25) {\huge$f_0$}; 
        \draw[very thick,maincolor] (0.5,3.75) -- (5.5,3.75); 
        \node at (6,3.75) {\huge$f_1$}; 
        \draw[very thick,maincolor] (0.5,0.75) -- (5.5,0.75); 
        \node at (6,0.75) {\huge$f_3$}; 
        \draw[very thick,maincolor] (0.5,0) -- (5.5,0); 
        \node at (6,0) {\huge$f_4$};
        \draw[very thick,maincolor] (0.5,2.25) -- (5.5,2.25); 
        \node at (6,2.25) {\huge$f_2$}; 

        \foreach \a/\b/\c in {0/0.2/0.8, 1/0/0,0/1/0,0.55/0.45/0}{
        \draw[very thick] (1,\a*6) -- (2,\a*6); 
        \draw[very thick] (2.5,\b*6) -- (3.5,\b*6); 
        \draw[very thick] (4,\c*6) -- (5,\c*6); 
        }
    
    \end{tikzpicture}}
    \end{subfigure}
    \caption{Example execution of the Ladder mechanism with $n=4$ voters and $m=3$ projects. The left panel shows the positions of the phantoms at $t=\frac{1}{2}$ (before normalization is reached) while the right panel shows them at $t=\frac{11}{12}$ (exactly when normalization is reached).}
    \label{fig:mechanism}
\end{figure}

We illustrate the Ladder mechanism in \Cref{fig:mechanism}. The example displayed in the figure has four voters with reports $(0,0.2,0.8), (1,0,0), (0,1,0), (0.55,0.45,0)$. Normalization is reached at $t=\frac{11}{12}$, returning the budget distribution $(\frac{5}{12}, \frac{5}{12}, \frac{1}{6})$.

\section{Upper Bounds} \label{sec:upper}

In this section, we present our main results, i.e., upper bounds for the overfunding and underfunding guarantees provided by the Ladder mechanism that are \emph{essentially} tight. In this context, we write \emph{essentially} in order to refer to the fact that there is a small gap of $\mathcal{O}(\frac{1}{n})$ between the upper and lower bounds only in the case when $n$ is odd. 

To prove the overfunding guarantee (\Cref{lem:upwards-deviation-egal}), we need the fact that any moving phantom mechanism is monotone. That is, if a single voter increases its report for a single project while it decreases its report for all other projects, then this project receives at least as much funding as in the original instance. 
Formally, a budget aggregation mechanism is \emph{monotone} if, for any two profiles $P$ and $P'$ for which there exists a voter $i_0$ and a project $j_0$ such that 
$(P_{ij})_{j \in [m]}=(P'_{ij})_{j \in [m]}$ for all $i \in [n] \setminus i_0$, and $P_{i_0,j_0} < P'_{i_0,j_0}$ while $P_{i_0,j} \geq P'_{i_0,j}$ for all $j \in [m] \setminus j_0$, it holds that $\mov(P)_{j_{0}} \leq \mov(P')_{j_{0}}$. 

\begin{lemma}
    Any moving phantom mechanism is monotone. 
\end{lemma}

\begin{proof}
    Let $t \in [0,1]$ ($t' \in [0,1]$, respectively) be the time at which mechanism \mov reaches normalization on profile $P$ ($P'$, respectively). If $t\leq t'$, then $\mov(P)_{j_0} \leq \mov(P')_{j_0}$ since phantoms and voters on project $j_0$ are all weakly higher for $P'$ than for $P$ at the time of normalization. If $t>t'$,
    then voters and phantoms are weakly lower for $P'$ than for $P$ for all $j \in [m] \setminus \{j_0\}$ at the time of normalization, implying $\mov(P)_j \geq \mov(P')_j$.
    By normalization, this implies $\mov(P)_{j_0} \leq \mov(P')_{j_0}$. 
\end{proof}

We are now ready to prove \Cref{lem:upwards-deviation-egal}. 

\begin{theorem}\label{lem:upwards-deviation-egal}
Let \mov be a proportional moving phantom mechanism. Then, \mov overfunds by at most $\alpha$, where $$\alpha(n,m) = \begin{cases}\frac{1}{4} & \text{ for } n,m \in \mathbb{N}, n \text{ even} \\ \frac{1}{4}\big(1 - \frac{1}{n^2}\big) & \text{ for } n,m \in \mathbb{N}, n \text{ odd.}\end{cases}$$ 
\end{theorem}

\begin{proof}
Consider some profile $P$ with normalization achieved at time $t^*$, and some project $j \in [m]$. Denote by $N^-$ the set of voters with $P_{ij}<\mov(P)_j$ and let $n^-=|N^-|$. Note that for every voter $i \in N^-$ there must exist a project $j_i$ with $P_{ij_i}>\mov(P)_{j_i}$, since votes and outputs are normalized. Starting from $P$, construct a profile $P'$ by, for every voter $i \in N^-$, changing $i$'s vote to be single minded on project $j_i$. Note that, holding the position of the phantoms fixed at $\{ f_k(t^*): k \in [n]_0 \}$, the median on every coordinate is (weakly) lower in $P'$ than in $P$, with the median on project $j$ being the same in the two profiles. So it might be the case that to achieve normalization in profile $P'$, we need to advance the phantoms to $\{ f_k(t'): k \in [n]_0 \}$ for some $t'>t^*$. Therefore, $\mov(P)_j \le \mov(P')_j$. Let us now construct a profile $P''$ by starting with $P'$ and, for every voter $i \not \in N^-$, setting their vote to be single-minded on project $j$. By monotonicity, $\mov(P'')_j \ge \mov(P')_j$. By proportionality of \mov, we have $\mov(P'')_j=1-\frac{n^-}{n}$. Combining the inequalities, we get $\mov(P)_j \le \mov(P')_j \le \mov(P'')_j=1-\frac{n^-}{n}$.

To complete the proof, note that 
$n-n^-$ voters have report $P_{ij} \ge \mov(P)_j$, by the definition of $N^-$.
Therefore,  $\overline{P}_j \ge (1-\frac{n^-}{n})\mov(P)_j$. We have
\begin{align*}
\mov(P)_j-\overline{P}_j &\le \mov(P)_j - \left(1-\frac{n^-}{n}\right)\mov(P)_j\\ 
&=\frac{n^-}{n}\mov(P)_j
\le \frac{n^-}{n}\left(1-\frac{n^-}{n}\right),
\end{align*}
which is at most $\frac{1}{4}$ when $n$ is even and at most $\frac{1}{4}\big(1 - \frac{1}{n^2}\big)$ when $n$ is odd. 
\end{proof}

We can easily verify that the Ladder mechanism satisfies proportionality: \citet[Section 5]{freeman2021truthful} argue that a moving phantom mechanism satisfies proportionality if there exists $t \in [0,1]$ such that $f_k(t)=1-\frac{k}{n}$ holds for all $k \in [n]_0$, which is the case for the Ladder mechanism when $t=1$. Hence, as an immediate corollary of \Cref{lem:upwards-deviation-egal}, we get that the overfunding guarantee of the Ladder mechanism is essentially optimal. We now turn to proving our main result, i.e., an essentially tight upper bound for the underfunding guarantee of the Ladder mechanism. We provide a proof sketch in \Cref{fig:main}.

\begin{figure}[t]
    \centering
    \scalebox{0.9}{
     \begin{tikzpicture}
        \draw[very thick] (0,0) -- (0,6);  
        \draw[very thick] (-.2,0) -- (0.2,0); 
        \draw[very thick] (-.2,6) -- (0.2,6); 
        \node at (-0.5,0) {$0$}; 
        \node at (-0.5,6) {$1$}; 
        \filldraw[fill=black!10,draw=none] (1,0) rectangle (2,6); 
        \filldraw[fill=black!10,draw=none] (2.5,0) rectangle (3.5,6); 
        \filldraw[fill=black!10,draw=none] (4,0) rectangle (5,6); 
        \node at (1.5,-0.5) {$1$};
        \node at (3,-0.5) {$2$};
        \node at (3.75,-0.5) {$\dots$};
        \node at (4.5,-0.5) {$m$};

        \filldraw[fill=color2!50,draw=none] (1,0.6) rectangle (2,6); 
        \node[rotate=90] at (1.5,4.5){(i) l.b. \# voters};
        \filldraw[fill=color3!50,draw=none] (2.5,1.2) rectangle (3.5,6); 
        \node[rotate=90] at (3,4.5){(iii) l.b. \# voters};
        \filldraw[fill=color3!50,draw=none] (4,0.6) rectangle (5,6); 
        \node[rotate=90] at (4.5,4.5){(iii) l.b. \# voters};

        \foreach \d/\e/\f in {0.1/0.2/0.1}{
        \draw[fill=white,thick] (1,\d*6-0.1) rectangle (2,\d*6+0.1); 
        \draw[fill=white,thick] (2.5,\e*6-0.1) rectangle (3.5,\e*6+0.1); 
        \draw[fill=white,thick] (4,\f*6-0.1) rectangle (5,\f*6+0.1); 
        }

        \draw[very thick,maincolor] (0.5,3) -- (5.5,3); 
        \node at (7,3) {(ii) u.b. $f_0(t^*)$}; 
        \draw[very thick,maincolor] (0.5,1.5) -- (5.5,1.5); 
        \draw[very thick,maincolor] (0.5,0) -- (5.5,0);         
        \draw[very thick,maincolor] (0.5,0) -- (5.5,0); 
        \draw[very thick,maincolor] (0.5,0) -- (5.5,0);

        \draw[very thick,draw=color2!50!black] (1,2.5) -- (2,2.5); 
        \node at (0.5,2.5) {\textcolor{color2!50!black}{$\overline{P}_1$}};         
    \end{tikzpicture}
}
    \caption{Proof sketch of \Cref{thm:main}: We assume for contradiction that the Ladder mechanism underfunds project $1$ by more than $\frac{1}{2}\big(1-\frac{1}{m}\big)$. We divide the proof into four steps: Using that the mean for project $1$ is high, in step (i), we derive a lower bound for the number of voters reporting a value in the interval $(\mov(P)_1,1]$ (indicated by green). (ii) 
    Since we know that the total number of phantoms and voters strictly above $\mov(P)_1$ is at most $n$,
    we derive an upper bound for the highest phantom at the point of normalization, i.e., $f_0(t^*)$. (iii) Building upon (ii), we can upper bound the number of phantoms within each interval $[\mov(P)_j,1]$ (indicated by orange) and thereby lower bound the number of voters reporting a value in the same interval. This in turn allows us to derive a lower bound on the mean of each project $j \neq 1$. (iv) Summing over all lower bounds on the mean implies a contradiction to the fact that the means sum up to $1$.}
    \label{fig:main}
\end{figure}
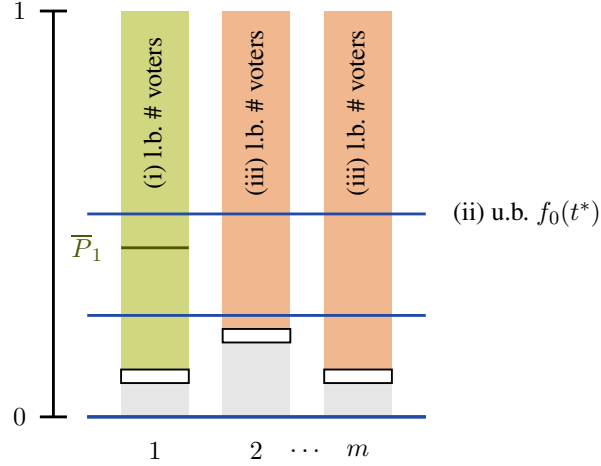

\begin{theorem} \label{thm:main}
    The Ladder mechanism underfunds by at most $\alpha$, where $$\alpha(n,m) = \frac{1}{2}\Big(1 -\frac{1}{m}\Big) \quad \text{ for all } n,m \in \mathbb{N}.$$
\end{theorem}

\begin{proof}
Let \mov be the Ladder mechanism and let $P$ be a preference profile. For the sake of contradiction, assume that there exists a project $j \in M$ with \begin{equation}
\overline{P}_j - \mov(P)_j>\frac{1}{2} - \frac{1}{2m}. \label{eq:contradiction-assumption}
\end{equation} We assume without loss of generality that $j=1$. 

We introduce the following notation: For simplicity, we write $a_j = \mov(P)_j$ for all $j \in M$. For a given project $j \in M$ and some interval $I \subseteq [0,1]$, we denote by $n_j(I)$ the number of agents within the interval $I$, i.e., $|\{i \in N \mid P_{ij} \in I\}|$. Similarly, we denote by $p(I)$ the number of phantoms in interval $I$, i.e., $|\{k \in [n]_0] \mid f_k(t^*) \in I\}|$, where $t^* \in [0,1]$ is some arbitrary point of normalization. 

\medskip 
\textbf{Step (i)} We start by showing that
\begin{equation}
    n_1((a_1,1]) \geq n \cdot \frac{\mu_1 - a_1}{1 - a_1}. \label{eq:claim1}
\end{equation}
This is because, given $n_1((a_1,1])$ voters with report strictly above $a_1$, the highest possible mean is attained when all of them report $1$ and the remaining voters report $a_1$. Formally, $$n_1((a_1,1]) \cdot 1 + (n- n_1((a_1,1])) \cdot a_1 \geq n \mu_1.$$ Rearranging this inequality yields \Cref{eq:claim1}.

\medskip

\textbf{Step (ii)} In this step, we derive an upper bound on the value of the highest phantom, i.e., $f_0(t^*)$. By definition of $a_1$ as the median on the first project, it holds that $n_1((a_1,1]) + p((a_1,1]) \leq n$, which yields an upper bound for the number of phantoms strictly above $a_1$. Formally, 
\begin{align*}
p((a_1,1]) &\leq n - n_1((a_1,1]) \\ 
& \stackrel{\cref{eq:claim1}}{\leq} n \cdot \frac{1-a_1}{1-a_1} - n \cdot \frac{\mu_1 - a_1}{1 - a_1}= n \cdot \frac{1-\mu_1}{1-a_1}. 
\end{align*}
Since any two consecutive phantoms are separated by a distance of at most $\frac{1}{n}$,
this yields an upper bound for the position of the highest phantom. Namely, the highest phantom is located at a position no greater than $a_1 + \frac{1-\mu_1}{1-a_1}$.

\medskip
\textbf{Step (iii)} We next derive lower bounds on the mean for any other project $j \in \{2,\dots,m\}$: 

\begin{equation}
\mu_j \geq \big(1-a_1 + a_j - \frac{1-\mu_1}{1-a_1} \big)a_j. \label{eq:lowerboundmean}
\end{equation}

As this bound clearly holds in the case that $a_j =0$, we can assume in the following case distinction that $a_j >0$.

\smallskip
\noindent \textbf{Case 1:} There is no phantom at $a_j$. 
We claim the following upper bound on the number of phantoms weakly above $a_j$: 
\begin{align*}
    p([a_j,1]) &\leq \Big\lceil n\big(a_1 + \frac{1-\mu_1}{1-a_1} - a_j \big) \Big\rceil  \\ & \leq n \big(a_1 + \frac{1-\mu_1}{1-a_1} -a_j \big) +1
\end{align*}
This bound holds because of the upper bound on the highest phantom and the fact that any two consecutive phantoms above $a_j$ have distance exactly $\frac{1}{n}$. The ceiling function comes from the fact that $a_j$ and the lowest phantom above $a_j$ may have smaller distance. 

\smallskip
\noindent \textbf{Case 2:} There is a phantom at $a_j$. We claim the following upper bound on the number of phantoms weakly above $a_j$: 
\[p([a_j,1]) \leq n \big(a_1 + \frac{1-\mu_1}{1-a_1} -a_j \big) + 1.\]
This bound holds because of the upper bound on the highest phantom and the fact that any two consecutive phantoms above $a_j$ have a distance exactly $\frac{1}{n}$. Moreover, since $a_j >0$, this minimum distance also holds for $a_j$ and the smallest phantom above $a_j$. The $+1$ comes from the phantom at $a_j$.

\smallskip

Note that we obtained the same bound in both cases. Now, since $n_j([a_j,1]) + p([a_j,1]) \geq n +1$ (as $a_j$ is the median of $2n+1$ values), we can use the above bound to derive a lower bound on the number of voters weakly above $a_j$:
\[n_j([a_j,1]) \geq n + 1 - p([a_j,1]) \geq n\big(1-a_1 - \frac{1-\mu_1}{1-a_1} + a_j \big).\]
The lowest mean for project $j$ is attained when $n_j([a_j,1])$ agents report $a_j$ and the remaining agents report $0$. Hence, we obtain the following lower bound on the mean: 
\[\mu_j \geq \big(1-a_1 + a_j - \frac{1-\mu_1}{1-a_1} \big)a_j.\]

\medskip
\textbf{Step (iv)} We complete the proof by summing over the mean of all projects. Using our bounds from step (iii), the assumption that project $1$ is underfunded, and that $\sum_{j=1}^ma_j=1$, we get: 

\begin{align*}
\mu_1 \! + \! \sum_{j = 2}^{m} \mu_2 \! & \stackrel{eq. (\ref{eq:lowerboundmean})}{\geq} \mu_1 + \sum_{j = 2}^{m} \big(1-a_1 + a_j - \frac{1-\mu_1}{1-a_1}\big)a_j \\ 
    & = \mu_1 + \left( 1-a_1 -\frac{1-\mu_1}{1-a_1} \right) \sum_{j = 2}^{m} a_j + \sum_{j = 2}^{m} a_j^2 \\ 
    & = 2\mu_1 -1 + (1-a_1)^2 + \sum_{j = 2}^{m} a_j^2 \\ 
    & \geq 2\mu_1 \! -1 + (1-a_1)^2 + (m-1)\! \left(\frac{1-a_1}{m-1}\right)^2 \\ 
    & \!\!\!\stackrel{eq.(\ref{eq:contradiction-assumption})}{>} 2\big( a_1 + \frac{1}{2}-\frac{1}{2m}\big) -1 + \frac{m}{m-1}(1-a_1)^2 \\ 
    & = 2a_1 - \frac{1}{m} + \frac{m}{m-1}(1-a_1)^2 \\ 
    & = \frac{m-1}{m} - 2(1-a_1) + \frac{m}{m-1}(1-a_1)^2 \! + 1 \\
    & = \Big(\sqrt{\frac{m-1}{m}} - \sqrt{\frac{m}{m-1}}(1-a_1)\Big)^2 \! + 1 \geq 1, 
\end{align*}
where the unlabeled inequality stems from the fact that $\sum_{j=2}^{m} a_j^2$ is the sum of convex functions, minimized when all values $a_j$ are equal. We obtained a contradiction to the fact that the sum of the means is $1$.
\end{proof}

\section{Implications for $\ell_1$ Distance}
\label{sec:l1}

\citet{caragiannis2022truthful} study the maximum 
$\ell_1$ distance between the output of any truthful budget aggregation mechanism and the mean.
More precisely, consider some budget aggregation mechanism $\mathcal{A}$. Then, for some real number\footnote{In this section, we focus on upper bounds that hold for specific values of $m$ and are independent of $n$, hence, for the sake of presentation we omit the parameterization of the upper bounds here.} $\alpha \in [0,2]$, $\mathcal{A}$ is said to be \emph{$\alpha$-approximate for $m$}, if $$\sum_{j \in [m]} |\mathcal{A}(P)_j - \overline{P}_j| \leq \alpha$$ holds for any instance $P$ with $m$ projects. \citet{caragiannis2022truthful} show that 
no moving phantom mechanism can achieve an approximation better than $1-\frac{1}{m}$ and that 
the Piecewise Uniform mechanism is $(\tfrac{2}{3} + \epsilon)$-approximate for $m=3$ and some $\epsilon \in [0,10^{-5}]$. Building upon our results from \Cref{sec:upper}, we are able to improve upon this result and show that the Ladder mechanism is $\frac{2}{3}$-approximate for $m=3$ as well as non-trivial upper bounds for larger $m$. The bounds follow from the following general result relating overfunding and underfunding to approximation guarantees.

\begin{lemma}
\label{prop:egal-to-l1}
    If a mechanism $\mathcal{A}$ overfunds by at most $\beta$ and underfunds by at most $\gamma$ (for fixed $m \in \mathbb{N}$ and all $n \in \mathbb{N}$), then $\mathcal{A}$ is $\alpha$-approximate for $m$, where 
    $$
    \alpha = 2 \cdot \max_{k \in [m-1]} \min \{ k\beta, (m-k)\gamma \}.
    $$
\end{lemma}

\begin{proof}
    Consider some profile $P$. Suppose that $\mathcal{A}(P)_{j} \ge \overline{P}_j$ for $k$ projects and $\mathcal{A}(P)_j < \overline{P}_j$ for $m-k$ projects. Then,
    $$
    \sum_{j: \mathcal{A}(P)_j \ge \overline{P}_j} |\mathcal{A}(P)_j - \overline{P}_j| \le k \beta,
    $$
    and analogously, 
    $$
    \sum_{j: \mathcal{A}(P)_j < \overline{P}_j} |\mathcal{A}(P)_j - \overline{P}_j| \le (m-k)\gamma.
    $$
    Moreover, since both $\mathcal{A}(P)$ and $\overline{P}$ are normalized, 
    $$\sum_{j: \mathcal{A}(P)_j \ge \overline{P}_j} |\mathcal{A}(P)_j - \overline{P}_j|=\sum_{j: \mathcal{A}(P)_j < \overline{P}_j} |\mathcal{A}(P)_j - \overline{P}_j|,$$
    which together with the bounds above implies
    $$\sum_{j \in [m]} |\mathcal{A}(P)_j - \overline{P}_j| \leq 2 \cdot \min \{ k\beta, (m-k)\gamma \}.$$
    
    Taking the maximum over all possible choices of $k$ yields the result. Note that $k=0$ is impossible, while $k=m$ leads to an $\ell_1$ distance to the mean of $0$.
\end{proof}

The desired result is derived by applying \Cref{thm:main,lem:upwards-deviation-egal} and \Cref{prop:egal-to-l1}. We summarize our results in Table~\ref{tab:l1-summary}. 
Note that directly applying this approach for $m > 6$ gives upper bounds that are larger than the trivial upper bound of 2, which is why \Cref{cor:last} applies only to $3 \le m \le 6$.

\begin{table}
\begin{center}
\begin{tabular}{ c | c | c | c }
$m$ & Lo. Bound & Up. Bound & Previous Up. Bound \\ \hline
 $3$ & $2/3$ & $2/3$ & $2/3+ \epsilon$ \\
 $4$ & $3/4$ & $1$ & $2$  \\  
 $5$ & $4/5$ & $3/2$ & $2$  \\   
  $6$ & $5/6$ & $5/3$ & $2$ \\ 
\end{tabular}
\caption{A summary of our results for the worst case $\ell_1$ distance from the mean. Lower bounds (holding for any moving phantom mechanism) and the previous $m=3$ upper bound are from the work of~\citet{caragiannis2022truthful}. Other previous upper bounds are trivial. In the previous $m=3$ upper bound, $\epsilon$ is some small constant no larger than $10^{-5}$.}
\label{tab:l1-summary}
\end{center}
\end{table}

\begin{restatable}{corollary}{lbounds} \label{cor:last}
The Ladder mechanism is $\frac{2}{3}$-approximate for $m=3$, $1$-approximate for $m=4$, $\frac{3}{2}$-approximate for~$m=5$, and $\frac{5}{3}$-approximate for $m=6$.
\end{restatable}

\begin{proof}
    The ladder mechanism overfunds, for any $n \in \mathbb{N}$ and $m \in \mathbb{N}$ by at most $\beta = \tfrac{1}{4}$ and underfunds for any $n \in \mathbb{N}$ and $m\in \mathbb{N}$ by at most $\gamma = \tfrac{1}{2}-\tfrac{1}{2m}$. For $m=3$, we have
    \begin{align*}
        2\min \{ k\beta, (m-k)\gamma \} &= 2\min \{ \tfrac{1}{4}, \tfrac{2}{3} \} = \tfrac{1}{2} \text{ for } k=1\\
        2\min \{ k\beta, (m-k)\gamma \} &= 2\min \{ \tfrac{1}{2}, \tfrac{1}{3} \} = \tfrac{2}{3} \text{ for } k=2,
    \end{align*}
   hence the ladder mechanism is $\tfrac{2}{3}$-approximate.

    For $m=4$, we have
    \begin{align*}
        2\min \{ k\beta, (m-k)\gamma \} &= 2\min \{ \tfrac{1}{4}, \tfrac{9}{8} \} = \tfrac{1}{2} \text{ for } k=1\\
        2\min \{ k\beta, (m-k)\gamma \} &= 2\min \{ \tfrac{1}{2}, \tfrac{3}{4} \} = 1 \text{ for } k=2 \\
        2\min \{ k\beta, (m-k)\gamma \} &= 2\min \{ \tfrac{3}{4}, \tfrac{3}{8} \} = \tfrac{3}{4} \text{ for } k=3,
    \end{align*}
    hence the ladder mechanism is $1$-approximate.

    For $m=5$, we have
    \begin{align*}
        2\min \{ k\beta, (m-k)\gamma \} &= 2\min \{ \tfrac{1}{4}, \tfrac{16}{10} \} = \tfrac{1}{2} \text{ for } k=1\\
        2\min \{ k\beta, (m-k)\gamma \} &= 2\min \{ \tfrac{1}{2}, \tfrac{12}{10} \} = 1 \text{ for } k=2\\
        2\min \{ k\beta, (m-k)\gamma \} &= 2\min \{ \tfrac{3}{4}, \tfrac{8}{10} \} = \tfrac{3}{2} \text{ for } k=3\\
        2\min \{ k\beta, (m-k)\gamma \} &= 2\min \{ 1, \tfrac{4}{10} \} = \tfrac{4}{5} \text{ for } k=4,
    \end{align*}
    hence the ladder mechanism is $\frac{3}{2}$-approximate.

    For $m=6$, we have
    \begin{align*}
        2\min \{ k\beta, (m-k)\gamma \} &= 2\min \{ \tfrac{1}{4}, \tfrac{25}{12} \} = \tfrac{1}{2} \text{ for } k=1\\
        2\min \{ k\beta, (m-k)\gamma \} &= 2\min \{ \tfrac{1}{2}, \tfrac{20}{12} \} = 1 \text{ for } k=2\\
        2\min \{ k\beta, (m-k)\gamma \} &= 2\min \{ \tfrac{3}{4}, \tfrac{15}{12} \} = \tfrac{3}{2} \text{ for } k=3\\
        2\min \{ k\beta, (m-k)\gamma \} &= 2\min \{ 1, \tfrac{10}{12} \} = \tfrac{5}{3} \text{ for } k=4\\
        2\min \{ k\beta, (m-k)\gamma \} &= 2\min \{ \tfrac{5}{4}, \tfrac{5}{12} \} = \tfrac{5}{6} \text{ for } k=5,
    \end{align*}
    hence the ladder mechanism is $\frac{5}{3}$-approximate.
\end{proof}

\section{Discussion}
\label{sec:discussion}

We introduce the notion of project fairness for the budget aggregation problem, defined by the maximum difference between the funding that a project receives and the funding that it would have received under the mean division of the budget. Our main technical contribution is to define the Ladder mechanism and show that it achieves essentially tight project fairness bounds. Additionally, our result yields a guarantee on the maximum $\ell_1$ distance between the output of the Ladder mechanism and the mean division, which is optimal for $m=3$ and the first non-trivial guarantees for $m \in \{4,5,6\}$.

Several open questions remain. Perhaps most intriguing is whether we can achieve better project-fairness guarantees with strategyproof mechanisms that are not moving phantom mechanisms (a project fairness lower bound of $\alpha(n,m)=\frac{1}{4}$ follows directly from an argument of \citet[Theorem 6]{caragiannis2022truthful}, but this is far from our upper bounds). Of course, resolving this question in the affirmative would require a resolution to the question of whether moving phantom mechanisms comprise the complete space of strategyproof mechanisms in this setting. It would also be interesting to characterize the class of optimal project-fair moving phantom mechanisms. Finally, it may be possible to tighten our analysis in Section~\ref{sec:l1}. We have made use only of the project fairness bounds locally on each project, but perhaps a more global analysis would do better.

\section*{Acknowledgments}
Ulrike Schmidt-Kraepelin was supported by the \emph{Deutsche Forschungsgemeinschaft} (under grant BR 4744/2-1), the \emph{Centro de Modelamiento Matemático (CMM)} (under grant FB210005, BASAL funds for center of excellence from ANID-Chile) and \emph{ANID-Chile} (grant ACT210005). Moreover, she was supported by the National Science Foundation under Grant No. DMS-1928930 and by the Alfred P. Sloan Foundation under grant G-2021-16778 while being in residence at the Simons Laufer Mathematical Sciences Institute (formerly MSRI) in Berkeley, California, during the Fall 2023 semester.

We thank an anonymous AAAI reviewer for catching a mistake in an earlier version of Proposition 1.

\newpage

\bibliography{literature}
\bibliographystyle{abbrvnat}

\end{document}